\documentclass[fleqn,11pt,twoside]{article}

\usepackage{amsthm,amsthm,amssymb, color, xcolor,epsfig, graphics, subfigure}

\usepackage{amsmath, graphicx, latexsym, lscape }

\usepackage{bm}

\def\DIS{\displaystyle}
\def\Z{\mathbb{Z}}
\def\Q{\mathbb{Q}}

\theoremstyle{plain}
\newtheorem{thm}{Theorem}[section]
\newtheorem{conj}[thm]{Conjecture}

\newtheorem{lem}[thm]{Lemma}
\newtheorem{prop}[thm]{Proposition}

\theoremstyle{definition}
\newtheorem{dfn}[thm]{Definition}

\numberwithin{equation}{section}


\makeatletter
\newcommand{\copyrightnote}[2]{{\renewcommand{\thefootnote}{}
 \footnotetext{\small\it
\begin{flushleft}
 \copyright \ #1   #2  
\end{flushleft}}}}

\newcommand{\Name}[1]{\begin{flushleft}
                       \LARGE \bf #1
                       \end{flushleft}\vspace{-3mm}}

\newcommand{\Author}[1]{\begin{flushleft}
                       \it #1 \end{flushleft}}

\newcommand{\Address}[1]{\begin{flushleft}
                       \it #1 \end{flushleft}}

\newcommand{\Date}[1]{\begin{flushleft}
                      \small  \it #1 \end{flushleft}}

\newcommand{\evenhead}{Author \ name}
\newcommand{\oddhead}{Article \ name}

\renewcommand{\@evenhead}{
\hspace*{-3pt}\raisebox{-15pt}[\headheight][0pt]{\vbox{\hbox to \textwidth
{\thepage \hfil \evenhead}\vskip4pt \hrule}}}
\renewcommand{\@oddhead}{
\hspace*{-3pt}\raisebox{-15pt}[\headheight][0pt]{\vbox{\hbox to \textwidth
{\oddhead \hfil \thepage}\vskip4pt\hrule}}}
\renewcommand{\@evenfoot}{}
\renewcommand{\@oddfoot}{}

\long\def\@makecaption#1#2{%
  \vskip\abovecaptionskip
  \sbox\@tempboxa{\small \textbf{#1.}\ \ #2}%
  \ifdim \wd\@tempboxa >\hsize
    {\small \textbf{#1.}\ \ #2}\par
  \else
    \global \@minipagefalse
    \hb@xt@\hsize{\hfil\box\@tempboxa\hfil}%
  \fi
  \vskip\belowcaptionskip}

\newcommand{\JNMPnumberwithin}[3][\arabic]{%
  \@ifundefined{c@#2}{\@nocounterr{#2}}{%
    \@ifundefined{c@#3}{\@nocnterr{#3}}{%
      \@addtoreset{#2}{#3}%
      \@xp\xdef\csname the#2\endcsname{%
        \@xp\@nx\csname the#3\endcsname .\@nx#1{#2}}}}%
}

\renewenvironment{proof}[1][\proofname]{\par
  \normalfont
  \topsep6\p@\@plus6\p@ \trivlist
  \item[\hskip\labelsep\textbf{%
    #1\@addpunct{.}}]\ignorespaces
}{%
  \qed\endtrivlist
}

\newcommand{\resetfootnoterule} {
  \renewcommand\footnoterule{%
  \kern-3\p@
  \hrule\@width.4\columnwidth
  \kern2.6\p@}
}

\renewcommand{\footnoterule}{}

\makeatother

\theoremstyle{definition}

\begin{document}

\renewcommand{\evenhead}{ {\LARGE\textcolor{blue!10!black!40!green}{{\sf \ \ \ ]ocnmp[}}}\strut\hfill 
N Okubo
}
\renewcommand{\oddhead}{ {\LARGE\textcolor{blue!10!black!40!green}{{\sf ]ocnmp[}}}\ \ \ \ \   
On extensions of $q$-discrete Painlev\'e I, II equations
}

\thispagestyle{empty}
\newcommand{\FistPageHead}[3]{
\begin{flushleft}
\raisebox{8mm}[0pt][0pt]
{\footnotesize \sf
\parbox{150mm}{{Open Communications in Nonlinear Mathematical Physics}\ \  \ {\LARGE\textcolor{blue!10!black!40!green}{]ocnmp[}}
\ \ Vol.6 (2026) pp
#2\hfill {\sc #3}}}\vspace{-13mm}
\end{flushleft}}

\FistPageHead{1}{\pageref{firstpage}--\pageref{lastpage}}{ \ \ Article}

\strut\hfill

\strut\hfill

\copyrightnote{The author(s). Distributed under a Creative Commons Attribution 4.0 International License}

\Name{Co-primeness preserving higher dimensional extension of $q$-discrete Painlev\'{e} I, II equations}

\Author{Naoto Okubo}

\Address{Graduate School of Mathematical Sciences, the University of Tokyo, 3-8-1 Komaba, Tokyo 153-8914, Japan}

\Date{Received May 6, 2025; Accepted August 16, 2026}

\setcounter{equation}{0}

\smallskip

\noindent
{\bf Citation format for this Article:}\newline
Naoto Okubo, Co-primeness preserving higher dimensional extension of $q$-discrete Painlev\'e I, II equations,
{\it Open Commun. Nonlinear Math. Phys.}, {\bf 6}, ocnmp:15632, \pageref{firstpage}--\pageref{lastpage}, 2026.

\strut\hfill

\noindent
{\bf The permanent Digital Object Identifier (DOI) for this Article:}\newline
{\it 10.46298/ocnmp.15632}

\strut\hfill

\begin{abstract}
\noindent 
We construct the $q$-discrete Painlev\'{e} I and II equations and their higher order analogues by virtue of periodic cluster algebras.
Using particular $k \times k$ exchange matrices, we show that the cluster algebras corresponding to $k=4$ and $5$ give the $q$-discrete Painlev\'{e} I and II equations respectively.
For $k \ge 6$, we obtain higher-order discrete equations that satisfy an integrability criterion, namely, the co-primeness property.
\end{abstract}

\label{firstpage}


\section{Introduction}

A cluster algebra introduced by Fomin and Zelevinsky is a commutative ring described by cluster variables and coefficients \cite{1,2}.
It has a wide range of connections with various fields of mathematics and theoretical physics. 
One of the interesting connections is its relation to integrable systems.
The cluster variables and coefficients obtained from a mutation of an initial seed satisfy certain difference equations, and some of them are related to discrete integrable equations.
In particular, Hone and Inoue \cite{3} and the author \cite{4,5} have shown that some periodic cluster algebras give the discrete Painlev\'{e} equations and their analogues of higher degrees.
An important property of the cluster algebra is its Laurent phenomenon, that is, the cluster variables determined from initial seeds are always expressed as Laurent polynomials of the initial variables.
These discrete Painlev\'{e} type equations of higher degrees also have Laurent phenomenon as is guaranteed by their construction. 
On the other hand, an important property of the discrete Painlev\'{e} equations is the singularity confinement property \cite{6}.
Hence we naturally expect that these discrete Painlev\'{e} type equations also have the singularity confinement property.
However, it is fairly difficult to check whether the discrete mappings of higher degrees have confined singularities or not due mainly to the increase of singularity patterns in the mappings with higher degrees.
Recently, with the aim of refining integrability criteria, Kanki et al. have proposed the `irreducibility' and the `co-primeness' properties to distinguish integrable mappings \cite{7}.
Actually the co-primeness property is regarded as an algebraic reinterpretation of singularity confinement.
Let us consider a discrete mapping with Laurent phenomenon. 
The mapping has the irreducibility property if every iterate is an irreducible Laurent polynomial of the initial variables. 
The equation satisfies the co-primeness condition if every pair of iterates is co-prime as Laurent polynomials.

In this article, we construct the Painlev\'{e} type equations from cluster algebras which include the $q$-Painlev\'{e} I and II equations, and show that they have the co-primeness property.
In the next section, we briefly summarize necessary notions about cluster algebras.
Generalized $q$-discrete Painlev\'{e} equations are constructed in section 3, and their co-primeness property is proved in section 4.
In Section 5, we examine the integrability of Painlevé type equations by investigating singularity confinement and degree growth.
Section 6 is devoted to the concluding remarks.

\section{Seed mutations}

In this section, we briefly explain the notion of cluster algebra \cite{2} which we use in the following sections.
Let $\bm{x}=(x_1,x_2,\dots ,x_N), \bm{y}=(y_1,y_2,\dots ,y_N)$ be $N$-tuple variables.
Each $x_i$ is called a cluster variable and each $y_i$ is called a coefficient.
Let $B=(b_{i,j})_{i,j=1}^N$ be a $N\times N$ integer skew-symmetric matrix.
$B$ is called an exchange matrix.
The triple $(B,\bm{x},\bm{y})$ is called a seed.
A mutation is a particular transformation of seeds.
\begin{dfn}\cite{2}
Let $\mu_k :(B,\bm{x},\bm{y})\longmapsto(B',\bm{x'},\bm{y'})\quad (k=1,2,\dots ,N)$ be the mutation at $k$, defined as follows.
\begin{itemize}
\item New exchange matrix $B'=(b_{i,j}')_{i,j=1}^N$ is defined from $B$ as:
\begin{equation}
b_{i,j}'=
\begin{cases}
-b_{i,j}\quad&(i=k\ \text{or}\ j=k)\\
b_{i,j}+b_{i,k}b_{k,j}\quad&(b_{i,k},b_{k,j}>0)\\
b_{i,j}-b_{i,k}b_{k,j}\quad&(b_{i,k},b_{k,j}<0)\\
b_{i,j}\quad&(\text{otherwise})
\end{cases}.
\end{equation}
\item New cluster variables $\bm{x'}=(x_1',x_2',\dots ,x_N')$ are defined from $B$ and $\bm{x}$ as:
\begin{equation}
x_i'=
\begin{cases}
\DIS \frac{1}{x_k}\left(\prod_{b_{k,j}>0}x_j^{b_{k,j}}
+\prod_{b_{k,j}<0}x_j^{-b_{k,j}}\right)\quad&(i=k)\\
x_i\quad &(i\neq k)
\end{cases}.
\end{equation}
\item New coefficients $\bm{y'}=(y_1',y_2',\dots ,y_N')$ are defined from $B$ and $\bm{y}$ as:
\begin{equation}
y_i'=
\begin{cases}
y_k^{-1}\quad&(i=k)\\
y_i(y_k^{-1}+1)^{-b_{k,i}}\quad&(b_{k,i}>0)\\
y_i(y_k+1)^{b_{k,i}}\quad&(b_{k,i}<0)\\
y_i\quad&(b_{k,i}=0)
\end{cases}.
\end{equation}
\end{itemize}
\end{dfn}
For any seed $t=(B,\bm{x},\bm{y})$, it holds that
\begin{equation}
\mu_k^2(t)=t.
\end{equation}
For any seed $t=(B,\bm{x},\bm{y})$ and $(i,j)$ such that $b_{i,j}=0$, it holds that
\begin{equation}
\mu_i\mu_j(t)=\mu_j\mu_i(t).
\end{equation}
The following theorem implies that the cluster algebras show the Laurent phenomenon.
\begin{thm}\label{LP}\cite{2}
Let $X(t)$ be the set of all the cluster variables obtained by iterative mutations to the seed $t=(B,\bm{x},\bm{y})$.
If $x\in X(t)$, then $x\in\mathbb{Z}[\bm{x}^\pm]$.
\end{thm}

\section{Generalized $q$-discrete Painlev\'e I, II equations from mutations}

For $k\geq 4$, we define $k\times k$ exchange matrix $B_k$ as
\begin{equation}
B_4=
\begin{bmatrix}
0 & -1 & 2 & -1 \\
1 & 0 & -3 & 2 \\
-2 & 3 & 0 & -1 \\
1 & -2 & 1 & 0
\end{bmatrix},
\end{equation}
\begin{equation}
B_5=
\begin{bmatrix}
0 & -1 & 1 & 1 & -1 \\
1 & 0 & -2 & 0 & 1 \\
-1 & 2 & 0 & -2 & 1 \\
-1 & 0 & 2 & 0 & -1 \\
1 & -1 & -1 & 1 & 0
\end{bmatrix},
\end{equation}
\begin{equation}
B_k=
\begin{bmatrix}
0 & -1 & 1 & & & & 1 & -1 \\
1 & 0 & -2 & 1 & & & -1 & 1 \\
-1 & 2 & 0 & -2 & \ddots & & & \\
 & -1 & 2 & \ddots & \ddots & \ddots & & \\
 & & \ddots & \ddots & \ddots & -2 & 1 & \\
 & & & \ddots & 2 & 0 & -2 & 1 \\
-1 & 1 & & & -1 & 2 & 0 & -1 \\
1 & -1 & & & & -1 & 1 & 0
\end{bmatrix}\quad(k\geq 6).
\end{equation}
These matrices $B_k$ have the following form of mutation-period:
\begin{equation}
B_k
\overset{\mu_1}{\longleftrightarrow}RB_kR^{-1}
\overset{\mu_2}{\longleftrightarrow}R^2B_kR^{-2}
\overset{\mu_3}{\longleftrightarrow}\cdots
\overset{\mu_k}{\longleftrightarrow}R^kB_kR^{-k}=B_k,
\end{equation}
where
\begin{equation}
R=\begin{bmatrix}
0 & & & 1 \\
1 & \ddots & & \\
 & \ddots & \ddots & \\
 & & 1 & 0
\end{bmatrix}.
\end{equation}
For the initial seed $(B_k,\bm{x},\bm{y})$, where $\bm{x}=(x_1,x_2,\dots,x_k)$ and $\bm{y}=(y_{0,1},y_{0,2},\dots,y_{0,k})$, new cluster variables $x_n$ and coefficients $y_{n,i}$ are defined as
\begin{equation}
\begin{aligned}
&(x_1,x_2,\dots,x_k)\\
\overset{\mu_1}{\longrightarrow}&(x_{k+1},x_2,\dots,x_k)\\
\overset{\mu_2}{\longrightarrow}&(x_{k+1},x_{k+2},\dots,x_k)\\
\overset{\mu_3}{\longrightarrow}&\cdots\\
\overset{\mu_k}{\longrightarrow}&(x_{k+1},x_{k+2},\dots,x_{2k})\\
\overset{\mu_1}{\longrightarrow}&(x_{2k+1},x_{k+2},\dots,x_{2k})\\
\overset{\mu_2}{\longrightarrow}&(x_{2k+1},x_{2k+2},\dots,x_{2k})\\
\overset{\mu_3}{\longrightarrow}&\cdots,
\end{aligned}
\end{equation}
\begin{equation}
\begin{aligned}
&(y_{0,1},y_{0,2},\dots,y_{0,k})\\
\overset{\mu_1}{\longrightarrow}&(y_{1,1},y_{1,2},\dots,y_{1,k})\\
\overset{\mu_2}{\longrightarrow}&(y_{2,1},y_{2,2},\dots,y_{2,k})\\
\overset{\mu_3}{\longrightarrow}&\cdots.
\end{aligned}
\end{equation}
The following proposition is readily obtained from the definition of mutation.
\begin{prop}
For any $n\in\Z$, the cluster variables $x_n$ satisfy the bilinear equation
\begin{equation}\label{x1}
x_{n+k}x_n=x_{n+k-1}x_{n+1}+x_{n+k-2}x_{n+2}.
\end{equation}
\end{prop}
If $k=4$ and $5$, then the above bilinear equations are called the Somos-4 sequence and the Somos-5 sequence respectively \cite{8}.
From Theorem \ref{LP}, we obtain the following proposition.
\begin{prop}\label{LP2}
For any $n\in\Z$, the cluster variables $x_n$ are in $\Z[x_1^{\pm},x_2^{\pm},\dots,x_k^{\pm}]$.
\end{prop}
From the definition of mutation, the coefficients $y_{n,i}$ satisfy the following relations:
\begin{equation}
\begin{cases}
y_{n,n}=y_{n-1,n}^{-1}\\
y_{n,n+1}=y_{n-1,n+1}(y_{n-1,n}+1)\\
y_{n,n+2}=y_{n-1,n+2}(y_{n-1,n}^{-1}+1)^{-2}\\
y_{n,n+3}=y_{n-1,n+3}(y_{n-1,n}+1)
\end{cases}
\quad(k=4),
\end{equation}
\begin{equation}
\begin{cases}
y_{n,n}=y_{n-1,n}^{-1}\\
y_{n,n+1}=y_{n-1,n+1}(y_{n-1,n}+1)\\
y_{n,n+2}=y_{n-1,n+2}(y_{n-1,n}^{-1}+1)^{-1}\\
y_{n,n+3}=y_{n-1,n+3}\\
\quad\vdots\\
y_{n,n+k-3}=y_{n-1,n+k-3}\\
y_{n,n+k-2}=y_{n-1,n+k-2}(y_{n-1,n}^{-1}+1)^{-1}\\
y_{n,n+k-1}=y_{n-1,n+k-1}(y_{n-1,n}+1)
\end{cases}
\quad(k\geq 5),
\end{equation}
where we consider that an index $i$ of the coefficients $y_{n,i}$ takes a value in $\mathbb{Z}/k\mathbb{Z}$.
By putting $y_n:=y_{n,n+1}$,
we obtain the following proposition.
\begin{prop}
For any $n\in\Z$, the coefficients $y_n$ satisfy
\begin{equation}\label{y1}
y_{n+k}y_n=\frac{(y_{n+k-1}+1)(y_{n+1}+1)}{(y_{n+k-2}^{-1}+1)(y_{n+2}^{-1}+1)}.
\end{equation}
\end{prop}

\subsection{Case of $k=2m\ (m=2,3,\dots)$}

We consider the equation \eqref{y1} for the case of $k=2m\ (m=2,3,\dots)$.
From deformation of the equation \eqref{y1}, we obtain
\begin{equation}\label{eq1}
\frac{y_{n+k}y_{n+k-1}(y_{n+2}+1)}{y_{n+2}y_{n+1}(y_{n+k-1}+1)}
=\frac{y_{n+k-1}y_{n+k-2}(y_{n+1}+1)}{y_{n+1}y_{n}(y_{n+k-2}+1)}.
\end{equation}
Hence we find
\begin{equation}
\frac{y_{n+k-1}y_{n+k-2}(y_{n+1}+1)}{y_{n+1}y_{n}(y_{n+k-2}+1)}=\alpha,
\end{equation}
where $\alpha$ is a constant.
Similarly we have
\begin{equation}
y_{n+k-1}y_{n+1}\prod_{i=1}^{k-3}\frac{y_{n+i+1}^2}{y_{n+i+1}+1}
=\alpha y_{n+k-2}y_{n}\prod_{i=1}^{k-3}\frac{y_{n+i}^2}{y_{n+i}+1},
\end{equation}
and 
\begin{equation}
y_{n+k-2}y_{n}\prod_{i=1}^{k-3}\frac{y_{n+i}^2}{y_{n+i}+1}=\beta\alpha^n,
\end{equation}
where $\beta$ is another constant.
Therefore, we obtain the following proposition.
\begin{prop}
If $k\geq 4$ is even, then the coefficients $y_n$ satisfy the difference equation
\begin{equation}\label{y2}
y_{n+k-2}y_{n}=\beta\alpha^n\prod_{i=1}^{k-3}\frac{y_{n+i}+1}{y_{n+i}^2}.
\end{equation}
\end{prop}
In particular, for $k=4$, this equation is the $q$-discrete Painlev\'e I equation \cite{9}:
\begin{equation}
y_{n+2}y_n=\beta\alpha^n\frac{y_{n+1}+1}{y_{n+1}^2}.
\end{equation}
Difference equation \eqref{y2} has a bilinear form.
\begin{prop}
If $k\geq 4$ is even and $x_n$ satisfy the bilinear equation
\begin{equation}\label{x2}
x_{n+k}x_n=z_n(x_{n+k-1}x_{n+1}+x_{n+k-2}x_{n+2}),
\end{equation}
where
\begin{equation}\label{z1}
z_n=\beta^{1/(k-3)}\alpha^{(2n-k+2)/(2k-6)},
\end{equation}
then
\begin{equation}
y_n:=\frac{x_{n+k-1}x_{n+1}}{x_{n+k-2}x_{n+2}}
\end{equation}
satisfy the difference equation \eqref{y2}.
\end{prop}
Note that the bilinear equation \eqref{x2} is a non-autonomous generalization of the equation \eqref{x1}.

\subsection{Case of $k=2m+1\ (m=2,3,\dots)$}

We consider the equation \eqref{y1} for the case of $k=2m+1\ (m=2,3,\dots)$.
From deformation of the equation \eqref{y1}, we obtain the relation \eqref{eq1}
and have
\begin{equation}
\frac{y_{n+k-1}y_{n+k-2}(y_{n+1}+1)}{y_{n+1}y_{n}(y_{n+k-2}+1)}=\alpha^2,
\end{equation}
where $\alpha$ is a constant.
Similarly we obtain
\begin{equation}
\frac{\prod_{i=0}^{2m-2}y_{n+i+2}}{\prod_{i=1}^{m-1}(y_{n+2i+1}+1)}
=\alpha^2\frac{\prod_{i=0}^{2m-2}y_{n+i}}{\prod_{i=1}^{m-1}(y_{n+2i-1}+1)},
\end{equation}
and
\begin{equation}
\frac{\prod_{i=0}^{2m-2}y_{n+i}}{\prod_{i=1}^{m-1}(y_{n+2i-1}+1)}=\beta\gamma^{(-1)^n}\alpha^n,
\end{equation}
where $\beta,\gamma$ are constants.
Therefore we find the following proposition.
\begin{prop}
If $k=2m+1\geq 5$ is odd, then the coefficients $y_n$ satisfy the difference equation
\begin{equation}\label{y3}
y_{n+2m-2}y_{n}=\beta\gamma^{(-1)^n}\alpha^n\frac{\prod_{i=1}^{m-1}(y_{n+2i-1}+1)}{\prod_{i=1}^{2m-3}y_{n+i}}.
\end{equation}
\end{prop}
By putting $f_n:=y_{2n}, g_n:=y_{2n+1}$,
 we obtain from \eqref{y3},
\begin{equation}
\begin{cases}
\DIS f_{n+m-1}f_n=\beta\gamma\alpha^{2n}\prod_{i=0}^{m-2}\frac{g_{n+i}+1}{g_{n+i}}\prod_{i=1}^{m-2}\frac{1}{f_{n+i}}\\
\DIS g_{n+m-1}g_n=\beta\gamma^{-1}\alpha^{2n+1}\prod_{i=1}^{m-1}\frac{f_{n+i}+1}{f_{n+i}}\prod_{i=1}^{m-2}\frac{1}{g_{n+i}}
\end{cases}.
\end{equation}
In particular, for $k=5\ (m=2)$, this equation is the $q$-discrete Painlev\'e II equation \cite{10}:
\begin{equation}
\begin{cases}
\DIS f_{n+1}f_n=\beta\gamma\alpha^{2n}\frac{g_{n}+1}{g_{n}}\\
\DIS g_{n+1}g_n=\beta\gamma^{-1}\alpha^{2n+1}\frac{f_{n+1}+1}{f_{n+1}}
\end{cases}.
\end{equation}
Difference equation \eqref{y3} has a bilinear form.
\begin{prop}
If $k=2m+1\geq 5$ is odd and $x_n$ satisfy the bilinear equation \eqref{x2} 
with
\begin{equation}\label{z2}
\begin{aligned}
z_{2n}&=(\beta/\gamma)^{1/(m-1)}\alpha^{(2n-m+1)/(m-1)},\\
z_{2n+1}&=(\beta\gamma)^{1/(m-1)}\alpha^{(2n-m+2)/(m-1)},
\end{aligned}
\end{equation}
then
\begin{equation}
y_n:=\frac{x_{n+k-1}x_{n+1}}{x_{n+k-2}x_{n+2}}
\end{equation}
satisfy the difference equation \eqref{y3}.
\end{prop}

\section{Co-primeness of generalized $q$-discrete Painlev\'e I, II equations}

We give some definitions necessary for the proof of co-primeness of the generalized $q$-discrete Painlev\'{e} equations.
\begin{dfn}\cite{7}
A Laurent polynomial $f\in R:=\Z[x_1^{\pm},\dots,x_n^{\pm}]$ is irreducible if, for every decomposition $f=gh\ (g,h\in R)$, at least one of $g$ or $h$ is a unit in $R$.
\end{dfn}
\begin{dfn}\label{dfn1}\cite{7}
Two Laurent polynomials $f_1,f_2$ are co-prime in $R:=\Z[x_1^{\pm},\dots,x_n^{\pm}]$ if the following condition is satisfied: we have decompositions $f_1=g_1h,f_2=g_2h$ in $R$, then $h$ must be a unit in $R$.
\end{dfn}
\begin{dfn}\cite{7}
Let $F:=\Q(y_1,\dots,y_n)$. Two rational functions $f_1,f_2$ are
co-prime in $F$ if the following condition is satisfied: let us express $f_1,f_2$ as $f_1=g_1/g_2,f_2=h_1/h_2$ where $g_1,g_2,h_1,h_2\in\Q[y_1^{\pm},\dots,y_n^{\pm}]$, $(g_1,g_2)$ and $(h_1,h_2)$ are co-prime pairs of polynomials.
Then, every pair of polynomials $(g_i,h_j)\ (i,j=1,2)$ is co-prime in the sense of Definition \ref{dfn1}.
\end{dfn}
The following theorem is the main result of the present article:
\begin{thm}\label{thm0}
For $k=4$, we define rational functions $y_n\in\Q(y_0,y_{1},\alpha,\beta)$ by the difference equation \eqref{y2}.
If $|n-n'|>2$, then $y_n$ and $y_{n'}$ are co-prime.

For $k= 5$, we define rational functions $y_n\in\Q(y_0,y_1,\alpha,\beta,\gamma)$ by the difference equation \eqref{y3}.
If $|n-n'|>3$, then $y_n$ and $y_{n'}$ are co-prime.
\end{thm}
Accordingly we have two conjectures.
\begin{conj}\label{thm1}
Let $k\geq 6$ be even.
By \eqref{y2}, we define rational functions $y_n\in\Q(y_0,\dots,y_{k-3},\alpha,\beta)$.
If $|n-n'|>k-2$, then $y_n$ and $y_{n'}$ are co-prime.
\end{conj}
\begin{conj}\label{thm2}
Let $k=2m+1\geq 7$ be odd.
By \eqref{y3}, we define rational functions $y_n\in\Q(y_0,\dots,y_{2m-3},\alpha,\beta,\gamma)$.
If $|n-n'|>k-2$, then $y_n$ and $y_{n'}$ are co-prime.
\end{conj}
We use the following lemma \cite{7}.
\begin{lem}\label{lem}
Let $\{p_1,p_2,\dots,p_n\}$ and $\{q_1,q_2,\dots,q_n\}$ be two sets of independent variables with the following properties:
\begin{equation}
p_i\in\Z[q_1^{\pm},q_2^{\pm},\dots,q_n^{\pm}],
\end{equation}
\begin{equation}
q_i\in\Z[p_1^{\pm},p_2^{\pm},\dots,p_n^{\pm}],
\end{equation}
\begin{equation}
q_i\ \text{is irreducible as an element of}\ \Z[p_1^{\pm},p_2^{\pm},\dots,p_n^{\pm}],
\end{equation}
for $i=1,2,\dots,n$.
Let us take an irreducible Laurent polynomial
\begin{equation}
f(p_1,p_2,\dots,p_n)\in\Z[p_1^{\pm},p_2^{\pm},\dots,p_n^{\pm}],
\end{equation}
and another Laurent polynomial
\begin{equation}
g(q_1,q_2,\dots,q_n)\in\Z[q_1^{\pm},q_2^{\pm},\dots,q_n^{\pm}],
\end{equation}
which satisfies
\begin{equation}
f(p_1,p_2,\dots,p_n)=g(q_1,q_2,\dots,q_n).
\end{equation}
In these settings, the function $g$ is decomposed as
\begin{equation}
g(q_1,q_2,\dots,q_n)=p_1^{r_1}p_2^{r_2}\dots p_n^{r_n}\cdot\tilde{g}(q_1,q_2,\dots,q_n),
\end{equation}
where $r_1,r_2,\dots,r_n\in\Z$ and $\tilde{g}(q_1,q_2,\dots,q_n)$ is irreducible in $\Z[q_1^{\pm},q_2^{\pm},\dots,q_n^{\pm}]$.
\end{lem}
\begin{prop}\label{prop1}
For $k\geq4$, we define Laurent polynomials $x_n\in\Z[x_1^{\pm},x_2^{\pm},\cdots,x_k^{\pm}]$ by the bilinear equation \eqref{x1}.
If $n\neq n'$, then $x_n$ and $x_{n'}$ are co-prime.
\end{prop}
Note that all $x_n$ are in $\Z[x_1^{\pm},x_2^{\pm},\cdots,x_k^{\pm}]$ from Proposition \ref{LP2}.
\begin{proof}
Let $L:=\Z[x_1^{\pm},\cdots,x_k^{\pm}]$ and define the sequence $(u_n)$ satisfying \eqref{x1} by substituting $x_1=\cdots=x_{k-1}=1$ and $x_k=t$:
\begin{equation}
u_{n+k}u_n=u_{n+k-1}u_{n+1}+u_{n+k-2}u_{n+2},
\end{equation}
\begin{equation}
u_1=\cdots=u_{k-1}=1, u_k=t.
\end{equation}
We also define a sequence $(v_n)$ (Fibonacci sequence):
\begin{equation}
v_n=v_{n-1}+v_{n-2},
\end{equation}
\begin{equation}
v_1=v_2=1.
\end{equation}
Clearly $x_{k+1}$ is irreducible in $L$. 
From Lemma \ref{lem}, we have
\begin{equation}\label{aux_eq1}
x_j=x_{k+1}^af_{irr}\quad(j=k+2,\cdots,2k+1),
\end{equation}
where $a\in\Z_{\geq 0}$ and $f_{irr}$ is irreducible in $L$.
By substituting $x_1=\cdots=x_{k-1}=1, x_k=t$, \eqref{aux_eq1} turns to
\begin{equation}
u_j=u_{k+1}^a\tilde{f}_{irr}=(t+1)^a\tilde{f}_{irr}\quad(j=k+2,\cdots,2k+1).
\end{equation}
Here $\tilde{f}_{irr}$ is a Laurent polynomial of $t$.
We can prove inductively
\begin{equation}
u_j=u_{j-1}+u_{j-2}=v_{j-k+1}t+v_{j-k}\quad(j=k+2,\cdots,2k-3).
\end{equation}
Since $v_i$ and $v_{i+1}$ are mutually co-prime, $u_j$ does not have a factor of $t+1$,
that implies $a=0$ and $x_j\ (j=k+2,\cdots,2k-3)$ is irreducible.

A straightforward calculation gives 
\begin{equation}
u_{2k-2}=u_{2k-3}+u_{2k-4}u_k
=u_{2k-3}+u_{2k-4}t
=v_{k-3}t^2+(v_{k-2}+v_{k-4})t+v_{k-3}.
\end{equation}
If we suppose that $u_{2k-2}$ has the following form:
\begin{equation}
u_{2k-2}=(c_1t+c_2)(t+1),
\end{equation}
then $c_1=v_{k-3}, c_1+c_2=v_{k-2}+v_{k-4}, c_2=v_{k-3}$, and $v_{k-3}=2v_{k-4}$ which is a contradiction.
Therefore $x_{2k-2}$ is irreducible.

The irreducibility of $x_{2k-1},x_{2k},x_{2k+1}$ is proved in a similar manner.

Next, from Lemma \ref{lem},  
\begin{equation}
x_j=x_{k+1}^af_{irr}=x_{k+2}^{b_{k+2}}\cdots x_{2k+1}^{b_{2k+1}}g_{irr}\quad(j=2k+2,\cdots,3k+1),
\end{equation}
where $a,b_i\in\Z_{\geq 0}$, and $f_{irr},g_{irr}$ are irreducible in $L$.
If we suppose that $x_j$ is not irreducible $(j=2k+2,\cdots,3k+1)$, we have 
\begin{equation}
x_j=x_{k+1}x_ih\quad(i=k+2,\cdots,2k+1).
\end{equation}
Here $h$ is a unit in $L$, that is, a monic Laurent monomial in $x_1,\cdots,x_k$.
Let $(\tilde{x}_n)$ be the sequence obtained by \eqref{x1} with initial condition $x_1=\cdots=x_k=1$, then we have
\begin{equation}
\tilde{x}_j=2\tilde{x}_i\quad(i=k+2,\cdots,2k+1).
\end{equation}
However, we can prove inductively
\begin{equation}\label{eq2}
\tilde{x}_n>2\tilde{x}_{n-1}\quad(n\geq k+3),
\end{equation}
which is a contradiction.
Therefore $x_j\ (j=2k+2,\cdots,3k+1)$ are irreducible.
\begin{itemize}
\item (Proof of \eqref{eq2}): \\
Since $\tilde{x}_{k+2}=3,\,\tilde{x}_{k+3}=7$,  
it holds that $\tilde{x}_{k+3}>2\tilde{x}_{k+2}$. 
For $n \ge k+3$, suppose that $\tilde{x}_{n}>2\tilde{x}_{n-1}$.
Then we have 
\begin{equation}
\begin{aligned}
\tilde{x}_{n+1}&=\frac{\tilde{x}_{n}\tilde{x}_{n-k+2}+\tilde{x}_{n-1}\tilde{x}_{n-k+3}}{\tilde{x}_{n-k+1}}\\
&>\frac{2(\tilde{x}_{n}\tilde{x}_{n-k+2}+\tilde{x}_{n-1}\tilde{x}_{n-k+3})}{\tilde{x}_{n-k+2}}\\
&>\frac{2(\tilde{x}_{n}\tilde{x}_{n-k+2}+2\tilde{x}_{n-1}\tilde{x}_{n-k+2})}{\tilde{x}_{n-k+2}}\\
&=2(\tilde{x}_{n}+2\tilde{x}_{n-1})\\
&>2\tilde{x}_n.
\end{aligned}
\end{equation}
Hence \eqref{eq2} holds by induction.
\end{itemize}

For $j \ge 3k+2$, we use Lemma \ref{lem} again and find
\begin{equation}\label{aux_eqq3}
x_j=x_{k+1}^af_{irr}=x_{k+2}^{b_{k+2}}\cdots x_{2k+1}^{b_{2k+1}}g_{irr}=x_{2k+2}^{c_{2k+2}}\cdots x_{3k+1}^{c_{3k+1}}h_{irr}\quad(j\geq 3k+2),
\end{equation}
where $a,b_i,c_i\in\Z_{\geq 0}$ and $f_{irr},g_{irr},h_{irr}$ are irreducible in $L$.
The relation \eqref{aux_eqq3} implies $a=0,\, b_{k+2}=\cdots =b_{2k+1}=0,\, c_{2k+2}=\cdots=c_{3k+1}=0$.
Hence $x_j$ is irreducible ($j \ge 3k+2$). 
Therefore $x_n$ is irreducible in $L$ for any $n \in \Z_{\ge 0}$.
Furthermore, from \eqref{eq2}, $x_n/x_{n'}$ ($n\neq n'$) is not a Laurent monomial.
Hence, $x_n$ and $x_{n'}$ are mutually prime. 
\end{proof}

\begin{prop}\label{prop2}
Let $z_n$ be given by \eqref{z1} or \eqref{z2} with $k\geq4$.
For initial variables $x_1,\cdots,x_k$, we define $x_n\in\Q(\alpha,\beta,\gamma)(x_1,\cdots,x_k)$ by \eqref{x2}.
Then $x_n$ and $x_{n'}$ are mutually co-prime for $n\neq n'$.
\end{prop}
\begin{proof}
From \eqref{z1} and \eqref{z2}, $\alpha=\beta=\gamma=1$ gives $z_n=1$. Then Proposition \ref{prop1} gives
Proposition \ref{prop2}.
\end{proof}

Finally we prove Theorem \ref{thm0}.
\begin{proof}[Proof of Theorem \ref{thm0}]
From Proposition \ref{prop2},  $x_j$ is an irreducible Laurent polynomial of $(x_1,x_2,x_3,x_4)$, and we easily find it is also irreducible in $L(x_1^\pm,x_2^\pm,y_0^\pm,y_1^\pm)$.
Hence for $|n-n'|>2$, $y_n$ and $y_{n'}$ are mutually co-prime in  $\Q (x_1,x_2,y_0,y_1,\alpha,\beta)$.
But $y_n$ and $y_{n'}$ depend only on $(y_0,y_1,\alpha,\beta)$ and the assertion is true for $k=4$.
Similarly we can prove in the case of $k=5$.
\end{proof}

For $k \ge 6$, a similar approach does not hold.
However, we note that the conjectures \ref{thm1} and \ref{thm2} are true in $\Q(x_1,x_2,...,x_k)$.

\section{Singularity confinement and degree growth}

In this section, we investigate the integrability of the difference equations \eqref{y2} and \eqref{y3} by examining singularity confinement \cite{6} and degree growth.
We illustrate this in detail using the difference equation \eqref{y2} with $k=6$ as an example.

\subsection{Detailed analysis for $k=6$}

When $k=6$, the difference equation \eqref{y2} becomes
\begin{equation}
y_{n+4}y_{n}=\beta\alpha^n\frac{(y_{n+3}+1)(y_{n+2}+1)(y_{n+1}+1)}{y_{n+3}^2y_{n+2}^2y_{n+1}^2}.
\end{equation}
First, we examine singularity confinement.
Among the initial values $(y_0,y_1,y_2,y_3)$, we take $y_0,y_1,y_2$ to be generic and consider the singularity arising when $y_3$ is assigned a particular value.
In this case, the initial values $(y_0,y_1,y_2,y_3)$ have three degrees of freedom.
Setting $y_3=0$ and iterating the equation, we obtain $y_4=\infty$, while $y_5$ becomes indeterminate.
To resolve this indeterminacy, we set $y_3=\epsilon$, iterate the equation, and then take the limit $\epsilon\to0$.
Computing the sequence $\{y_n\}$, we obtain
\begin{equation}\label{k60s}
\{y_0,y_1,y_2,0,\infty^2,F_1,F_2,0,F_3\}.
\end{equation}
Here, $\infty^2$ indicates a quantity that diverges as $\epsilon^{-2}$ as $\epsilon\to0$, and $F_1,F_2,F_3$ are rational functions of the initial values $y_0,y_1,y_2$.
At the next step, $y_9$ is determined by $(F_1,F_2,0,F_3)$.
By computing the Jacobian determinant of $F_1,F_2,F_3$ with respect to $y_0,y_1,y_2$, we find that it is not identically zero, and hence $F_1,F_2,F_3$ are algebraically independent.
Thus, the three degrees of freedom of the initial values $(y_0,y_1,y_2,0)$ are recovered in $(F_1,F_2,0,F_3)$.
This property is called singularity confinement and is known as one of the criteria for integrability.
Similarly, setting the initial value $y_3$ to $-1$ or $\infty$ and iterating the equation, we obtain
\begin{equation}\label{k6-1}
\{y_0,y_1,y_2,-1,0,\infty,\bullet,\infty,0,-1,\bullet,\bullet,\bullet\},
\end{equation}
\begin{equation}\label{k6infs}
\{y_0,y_1,y_2,\infty,0,\infty,\bullet,\bullet,\bullet\}.
\end{equation}
Here, $\bullet$ denotes a rational function of the initial values $y_0,y_1,y_2$.
In both cases, the three degrees of freedom of the initial values are recovered, showing that the singularities are confined.

Another well-known method for testing integrability is to examine algebraic entropy \cite{11}.
Algebraic entropy measures the rate of degree growth of the rational expressions arising under iteration of a map.
Let the initial values $(y_0,y_1,y_2,y_3)$ be generic, and let $d_n$ denote the degree of $y_n$ with respect to $y_3$, where $y_n$ is regarded as a rational function of the initial values.
Then, the algebraic entropy is defined by $\lambda:=\lim_{n\to\infty}(\log d_n)/n$.
Vanishing algebraic entropy is known as one of the criteria for integrability.
Several methods for computing algebraic entropy are known.
In \cite{12}, the degree is computed exactly by exploiting the co-primeness property of the difference equation.
Another approach is Halburd's method \cite{13}, which determines the degree exactly by analyzing singularity patterns.
In \cite{14, 15}, Halburd's method is applied to higher-order difference equations to compute their degrees.
When an exact computation of the degree is not required and one is only interested in determining whether the algebraic entropy vanishes, the simpler express method \cite{16, 17} can be used.
In \cite{18}, Halburd's method is extended to partial difference equations.
A rigorous and detailed explanation of how to compute the degree of a difference equation exactly using Halburd's method is also given in \cite{18}.
In the following discussion, we refer to \cite{18} in computing the degree $d_n$ of the difference equation \eqref{y2} exactly.

To determine the degree $d_n$ exactly using Halburd's method, we need to examine when $y_n$ takes the singular values $0$, $-1$, and $\infty$ under iteration of the difference equation.
The singularity pattern \eqref{k60s} is part of the following periodic pattern with period $22$:
\begin{equation}\label{k60}
\{\boxed{y_0,y_1,y_2,0},\infty^2,\boxed{\bullet,\bullet,0,\bullet},\infty^2,\boxed{\bullet,0,\bullet,\bullet},\infty^2,\boxed{0,\bullet,\bullet,\bullet},\infty,0,\infty,\boxed{\bullet,\bullet,\bullet,0},\dots\}.
\end{equation}
It is difficult to compute this entire periodic pattern directly from the initial values $(y_0,y_1,y_2,0)$.
However, as mentioned above, the three rational functions appearing in $\boxed{\bullet,\bullet,0,\bullet}$ are algebraically independent.
Therefore, the continuation of this pattern can be obtained by setting the initial values to $(y_0,y_1,0,y_3)$ and computing the subsequent iterates starting from $y_4$.
Similarly, it can be verified that the three rational functions appearing in $\boxed{\bullet,0,\bullet,\bullet}$ and $\boxed{0,\bullet,\bullet,\bullet}$ are also algebraically independent.
Hence, the computation can be restarted from an intermediate point in the pattern.
Moreover, the three rational functions appearing in the final $\boxed{\bullet,\bullet,\bullet,0}$ are also algebraically independent.
The above calculations confirm the periodic pattern \eqref{k60}.
The singularity pattern \eqref{k6infs} is also part of the periodic pattern \eqref{k60}.
By shifting the starting point of the periodic pattern \eqref{k60}, it can be written as
\begin{equation}\label{k6inf}
\{y_0,y_1,y_2,\infty,0,\infty,\boxed{\bullet,\bullet,\bullet,0},\infty^2,\boxed{\bullet,\bullet,0,\bullet},\infty^2,\boxed{\bullet,0,\bullet,\bullet},\infty^2,\boxed{0,\bullet,\bullet,\bullet},\infty,\dots\}.
\end{equation}
In the singularity pattern \eqref{k6-1}, rational functions of the initial values $y_0,y_1,y_2$ continue to appear repeatedly after $y_{13}$, and hence the pattern is not periodic.
If Conjecture \ref{thm1} is true, the non-periodicity of this pattern can be proved rigorously.
Since the last zero or pole arising from $y_3=-1$ occurs at $y_8$ and $13-8>4$, Conjecture \ref{thm1} implies that no such zero or pole can occur in $y_{13},y_{14},\dots$.
Hence, this singularity pattern is not cyclic.

Let $(y_0,y_1,y_2,y_3)$ be the initial values, and express $y_n$ as a rational function of the initial values.
Regarding $y_n$ as a rational function of $y_3$, we write $y_n=f_n(y_3)$.
Then, $d_n=\deg f_n$.
It is well known that, for any $\alpha\in\mathbb{P}^1(\mathbb{C})$, the degree $d_n$ is equal to the number of preimages of $\alpha$ under $f_n$, counted with multiplicity.
The idea of Halburd's method is to compute the degree by counting the preimages of the singular values $0$, $-1$, and $\infty$ appearing in the singularity patterns.
Let $y^*\in\mathbb{P}^1(\mathbb{C})$.
Suppose that $f_0(y^*),f_1(y^*),\dots,f_{n-1}(y^*)$ take generic values and that $f_n(y^*)=-1$.
Then, from the singularity pattern \eqref{k6-1}, we obtain
\begin{equation}\label{pat}
\{y_0,y_1,y_2,y_3=y^*,\bullet,\dots,\bullet,y_n=-1,0,\infty,\bullet,\infty,0,-1,\bullet,\bullet,\dots\}.
\end{equation}
Let $Z_n$ denote the number of $y^*\in\mathbb{P}^1(\mathbb{C})$ that generate this pattern.
Namely, we define
\begin{equation}
Z_n=\#\{y^*\in\mathbb{P}^1(\mathbb{C})\mid y_3=y^*\ \text{generates the singularity pattern \eqref{pat}}\}.
\end{equation}
We first determine the number of preimages of $0$ under $f_n$.
Considering when $y_n=0$ can occur, there are the following four possibilities.
\begin{itemize}
\item
When $f_0(y^*),f_1(y^*),\dots,f_{n-2}(y^*)$ take generic values and $f_{n-1}(y^*)=-1$, that is, when
\begin{equation}
\{y_0,y_1,y_2,y_3=y^*,\bullet,\dots,\bullet,y_{n-1}=-1,0,\infty,\bullet,\infty,0,-1,\bullet,\bullet,\dots\}.
\end{equation}
The number of $y^*\in\mathbb{P}^1(\mathbb{C})$ that generate this pattern is $Z_{n-1}$.

\item
When $f_0(y^*),f_1(y^*),\dots,f_{n-6}(y^*)$ take generic values and $f_{n-5}(y^*)=-1$, that is, when
\begin{equation}
\{y_0,y_1,y_2,y_3=y^*,\bullet,\dots,\bullet,y_{n-5}=-1,0,\infty,\bullet,\infty,0,-1,\bullet,\bullet,\dots\}.
\end{equation}
The number of $y^*\in\mathbb{P}^1(\mathbb{C})$ that generate this pattern is $Z_{n-5}$.

\item
When the initial values are $(y_0,y_1,y_2,y_3=0)$, it follows from the periodic pattern \eqref{k60} that if $n\equiv3,7,11,15,20 \pmod{22}$ then $y_3=0$ is a preimage of $0$ under $f_n$.

\item
When the initial values are $(y_0,y_1,y_2,y_3=\infty)$, it follows from the periodic pattern \eqref{k6inf} that if $n\equiv4,9,13,17,21 \pmod{22}$ then $y_3=\infty$ is a preimage of $0$ under $f_n$.

\end{itemize}
From the above, we obtain
\begin{equation}\label{k60d}
d_n=\#\{y^*\in\mathbb{P}^1(\mathbb{C})\mid f_n(y^*)=0\}=Z_{n-1}+Z_{n-5}+\phi_n.
\end{equation}
Here, $\phi_n$ is a periodic function with period $22$, defined by
\begin{equation}
\phi_n=
\begin{cases}
1&(n\equiv3,4,7,9,11,13,15,17,20,21 \pmod{22})\\
0&(\text{otherwise})
\end{cases}.
\end{equation}
Similarly, by counting the preimages of $-1$ and $\infty$ under $f_n$, we obtain
\begin{equation}\label{k6-1d}
d_n=\#\{y^*\in\mathbb{P}^1(\mathbb{C})\mid f_n(y^*)=-1\}=Z_n+Z_{n-6}
\end{equation}
and
\begin{equation}
d_n=\#\{y^*\in\mathbb{P}^1(\mathbb{C})\mid f_n(y^*)=\infty\}=Z_{n-2}+Z_{n-4}+\psi_n,
\end{equation}
respectively.
Here, $\psi_n$ is a periodic function with period $22$, defined by
\begin{equation}
\psi_n=
\begin{cases}
2&(n\equiv4,9,10,14,15,20 \pmod{22})\\
1&(n\equiv3,5,19,21 \pmod{22})\\
0&(\text{otherwise})
\end{cases}.
\end{equation}
From \eqref{k60d} and \eqref{k6-1d}, we obtain the following recurrence relation for $Z_n$:
\begin{equation}
Z_n=Z_{n-1}+Z_{n-5}-Z_{n-6}+\phi_n.
\end{equation}
Solving this recurrence relation with the initial conditions $Z_{-2}=Z_{-1}=\dots=Z_2=0,\,Z_3=1$, we obtain
\begin{equation}
Z_n=\frac{n^2+4n}{22}+h_n.
\end{equation}
Here, $h_n$ is a periodic function with period $22$, defined by
\begin{equation}
\begin{aligned}
(h_0,h_1,\dots,h_{21})=\frac{1}{22}(&0, -5, -12, 1, 12, -1, -16, -11, -8, 15, -8, \\&{-11}, -16, -1, 12, 1, -12, -5, 0, 3, 4, 3).
\end{aligned}
\end{equation}
It follows from \eqref{k6-1d} that the degree $d_n$ is given by
\begin{equation}
d_n=\frac{n^2-2n}{11}+g_n.
\end{equation}
Here, $g_n$ is a periodic function with period $22$, defined by
\begin{equation}
\begin{aligned}
(g_0,g_1,\dots,g_{21})=\frac{1}{11}(&0, 1, 0, 8, 14, 7, -2, -2, -4, 14, 8, 0, \\&{-10}, 0, 8, 14, -4, -2, -2, 7, 14, 8).
\end{aligned}
\end{equation}
Since $d_n$ is expressed as the sum of a quadratic polynomial in $n$ and a periodic function, the algebraic entropy is zero.

\subsection{The case of even $k\geq4$}

Next, we consider the difference equation \eqref{y2} for even $k\geq 4$.
Let $(y_0,y_1,\dots,y_{k-3})$ be the initial values, and express $y_n$ as a rational function of the initial values.
Regarding $y_n$ as a rational function of $y_{k-3}$, we write $y_n=f_n(y_{k-3})$ and determine the degree $d_n=\deg f_n$.
When $k=4$, there are two singularity patterns: a non-periodic singularity pattern
\begin{equation}
\{y_0,-1,0,\infty^2,0,-1,\bullet,\bullet,\dots\}
\end{equation}
and a singularity pattern with period $8$
\begin{equation}\label{k40}
\{\boxed{y_0,0},\infty^2,\boxed{0,\bullet},\infty,0,\infty,\boxed{\bullet,0},\dots\}.
\end{equation}
When $k=8$, there are two singularity patterns: a non-periodic singularity pattern
\begin{equation}
\{y_0,y_1,y_2,y_3,y_4,-1,0,\infty,\bullet,\bullet,\bullet,\infty,0,-1,\bullet,\bullet,\dots\}
\end{equation}
and a singularity pattern with period $44$
\begin{equation}\label{k80}
\begin{aligned}
\{&\boxed{y_0,y_1,y_2,y_3,y_4,0},\infty^2,\boxed{\bullet,\bullet,\bullet,\bullet,0,\bullet},\infty^2,\boxed{\bullet,\bullet,\bullet,0,\bullet,\bullet},\infty^2,\boxed{\bullet,\bullet,0,\bullet,\bullet,\bullet},\infty^2,\\
&\boxed{\bullet,0,\bullet,\bullet,\bullet,\bullet},\infty^2,\boxed{0,\bullet,\bullet,\bullet,\bullet,\bullet},\infty,0,\infty,\boxed{\bullet,\bullet,\bullet,\bullet,\bullet,0},\dots\}.
\end{aligned}
\end{equation}
Based on the results for $k=4,6,8$, we conjecture that, for even $k\geq6$, there is a non-periodic singularity pattern
\begin{equation}\label{k2m-1}
\{y_0,y_1,\dots,y_{k-4},-1,0,\infty,\underbrace{\bullet,\dots,\bullet}_{k-5},\infty,0,-1,\bullet,\bullet,\dots\},
\end{equation}
and that, for even $k\geq4$, there is a singularity pattern with period $L:=k^2-3k+4$
\begin{equation}\label{k2m0}
\begin{aligned}
\{&\boxed{y_0,y_1,\dots,y_{k-4},0},\infty^2,\underbrace{\boxed{\bullet,\dots,\bullet,0,\bullet}}_{k-2},\infty^2,\underbrace{\boxed{\bullet,\dots,\bullet,0,\bullet,\bullet}}_{k-2},\infty^2,\dots,\\
&\underbrace{\boxed{\bullet,0,\bullet,\dots,\bullet}}_{k-2},\infty^2,\underbrace{\boxed{0,\bullet,\dots,\bullet}}_{k-2},\infty,0,\infty,\underbrace{\boxed{\bullet,\dots,\bullet,0}}_{k-2},\dots\}.
\end{aligned}
\end{equation}
The singularity pattern starting from the initial values $(y_0,y_1,\dots,y_{k-4},\infty)$ is contained in the latter periodic pattern.
We define $Z_n$ by
\begin{equation}
Z_n=\#\{y^*\in\mathbb{P}^1(\mathbb{C})\mid y_{k-3}=y^*\ \text{generates the singularity pattern \eqref{pat2m}}\},
\end{equation}
\begin{equation}\label{pat2m}
\{y_0,y_1,\dots,y_{k-3}=y^*,\bullet,\dots,\bullet,y_n=-1,0,\infty,\underbrace{\bullet,\dots,\bullet}_{k-5},\infty,0,-1,\bullet,\bullet,\dots\}.
\end{equation}
By counting the preimages of the singular values $0,-1,\infty$ under $f_n$, we obtain
\begin{align}
d_n&=\#\{y^*\in\mathbb{P}^1(\mathbb{C})\mid f_n(y^*)=0\}=Z_{n-1}+Z_{n-k+1}+\phi_n,
\label{k2m0d}\\
d_n&=\#\{y^*\in\mathbb{P}^1(\mathbb{C})\mid f_n(y^*)=-1\}=Z_n+Z_{n-k},
\label{k2m-1d}\\
d_n&=\#\{y^*\in\mathbb{P}^1(\mathbb{C})\mid f_n(y^*)=\infty\}=Z_{n-2}+Z_{n-k+2}+\psi_n,
\end{align}
where $\phi_n$ and $\psi_n$ are periodic functions with period $L$, defined by
\begin{align}
\phi_n&=
\begin{cases}
1&(n\equiv k-2,(k-1)(k-2),\\
&\hspace{8mm}\ell(k-2)-1\ (\ell=1,\dots,k-2),\\
&\hspace{8mm}\ell'(k-2)+1\ (\ell'=2,\dots,k-1) \pmod{L})\\
0&(\text{otherwise})
\end{cases},\\
\psi_n&=
\begin{cases}
2&(n\equiv\ell(k-1)-1\ (\ell=1,\dots,k-3),\\
&\hspace{8mm}\ell'(k-1)\ (\ell'=2,\dots,k-2) \pmod{L})\\
1&(n\equiv k-3,k-1,\\
&\hspace{8mm}(k-2)(k-1)-1,(k-2)(k-1)+1 \pmod{L})\\
0&(\text{otherwise})
\end{cases}.
\end{align}
It follows from \eqref{k2m0d} and \eqref{k2m-1d} that $Z_n$ satisfies the recurrence relation
\begin{equation}
\begin{cases}
Z_n=Z_{n-1}+Z_{n-k+1}-Z_{n-k}+\phi_n\\
Z_{-2}=Z_{-1}=\dots=Z_{k-4}=0,\,Z_{k-3}=1
\end{cases}.
\end{equation}
Solving this recurrence relation, we obtain
\begin{equation}
Z_n=\frac{n^2+4n}{L}+(\text{periodic function with period $L$}).
\end{equation}
Hence, by \eqref{k2m-1d}, the degree $d_n$ is given by
\begin{equation}
d_n=\frac{2n^2+(8-2k)n}{L}+g_n,
\end{equation}
where $g_n$ is a periodic function with period $L$.
This formula is proved for $k=4,6,8$.
For even $k\geq10$, the formula follows provided that the singularity patterns \eqref{k2m-1} and \eqref{k2m0} are correct.
Since $d_n$ is expressed as the sum of a quadratic polynomial in $n$ and a periodic function, the algebraic entropy is zero.
In the case $k=4$, $g_n$ has period $L/2=4$ and is defined by
\begin{equation}
(g_0,g_1,g_2,g_3)=\frac{1}{4}(0, 3, 4, 3).
\end{equation}
For $k=8$, $g_n$ has period $L=44$ and is defined by
\begin{equation}
\begin{aligned}
(g_0,g_1,\dots,g_{43})=\frac{1}{22}(&0, 3, 4, 3, 0, 17, 32, 23, 12, -1, -16, -11, -8, 37, 36, 11, -16, \\&{-23}, -32, 1, 32, 39, 0, -19, -40, -19, 0, 39, 32, 1, -32, \\&{-23}, -16, 11, 36, 37, -8, -11, -16, -1, 12, 23, 32, 17).
\end{aligned}
\end{equation}

\subsection{The case of odd $k\geq5$}

Next, we consider the singularity patterns of \eqref{y3} and the degree $d_n$ for odd $k\geq 5$.
Let $(y_0,y_1,\dots,y_{k-4})$ be the initial values, and express $y_n$ as a rational function of the initial values.
Regarding $y_n$ as a rational function of $y_{k-4}$, we write $y_n=f_n(y_{k-4})$ and determine the degree $d_n=\deg f_n$.
In this case, there are two singularity patterns: a non-periodic singularity pattern
\begin{equation}\label{k2m+1-1}
\{y_0,y_1,\dots,y_{k-5},-1,0,\infty,\underbrace{\bullet,\dots,\bullet}_{k-5},\infty,0,-1,\bullet,\bullet,\dots\}
\end{equation}
and a singularity pattern with period $L:=k^2-3k+4$
\begin{equation}\label{k2m+10}
\begin{aligned}
\{&\boxed{y_0,y_1,\dots,y_{k-5},0},\infty,\infty,\underbrace{\boxed{\bullet,\dots,\bullet,0,\bullet}}_{k-3},\infty,\infty,\underbrace{\boxed{\bullet,\dots,\bullet,0,\bullet,\bullet}}_{k-3},\infty,\infty,\dots,\\
&\underbrace{\boxed{\bullet,0,\bullet,\dots,\bullet}}_{k-3},\infty,\infty,\underbrace{\boxed{0,\bullet,\dots,\bullet}}_{k-3},\infty,\bullet,0,\infty,\underbrace{\bullet,\dots,\bullet}_{k-5},\infty,0,\bullet,\infty,\underbrace{\boxed{\bullet,\dots,\bullet,0}}_{k-3},\dots\}.
\end{aligned}
\end{equation}
These singularity patterns can be verified for $k=5,7,9$, whereas for $k\geq11$ they are conjectural.
The singularity pattern starting from the initial values $(y_0,y_1,\dots,y_{k-5},\infty)$ is contained in \eqref{k2m+10}.
We note that, in the case $k=5$, the singularity pattern \eqref{k2m+10} takes the form
\begin{equation}
\{\boxed{y_0,0},\infty,\infty,\boxed{0,\bullet},\infty,\bullet,0,\infty,\infty,0,\bullet,\infty,\boxed{\bullet,0},\dots\},
\end{equation}
and hence has period $L/2=7$, which is half of $L$.
We define $Z_n$ by
\begin{equation}
Z_n=\#\{y^*\in\mathbb{P}^1(\mathbb{C})\mid y_{k-4}=y^*\ \text{generates the singularity pattern \eqref{pat2m+1}}\},
\end{equation}
\begin{equation}\label{pat2m+1}
\{y_0,y_1,\dots,y_{k-4}=y^*,\bullet,\dots,\bullet,y_n=-1,0,\infty,\underbrace{\bullet,\dots,\bullet}_{k-5},\infty,0,-1,\bullet,\bullet,\dots\}.
\end{equation}
By counting the preimages of the singular values $0,-1,\infty$ under $f_n$, we obtain
\begin{align}
d_n&=\#\{y^*\in\mathbb{P}^1(\mathbb{C})\mid f_n(y^*)=0\}=Z_{n-1}+Z_{n-k+1}+\phi_n,
\label{k2m+10d}\\
d_n&=\#\{y^*\in\mathbb{P}^1(\mathbb{C})\mid f_n(y^*)=-1\}=Z_n+Z_{n-k},
\label{k2m+1-1d}\\
d_n&=\#\{y^*\in\mathbb{P}^1(\mathbb{C})\mid f_n(y^*)=\infty\}=Z_{n-2}+Z_{n-k+2}+\psi_n,
\end{align}
where $\phi_n$ and $\psi_n$ are periodic functions with period $L$, defined by
\begin{align}
\phi_n&=
\begin{cases}
1&(n\equiv k-2,2(k-2),(k-2)^2-1,(k-1)(k-2)-1,\\
&\hspace{8mm}\ell(k-2)-2\ (\ell=1,\dots,k-3),\\
&\hspace{8mm}\ell'(k-2)+1\ (\ell'=3,\dots,k-1) \pmod{L})\\
0&(\text{otherwise})
\end{cases},\\
\psi_n&=
\begin{cases}
2&(n\equiv\ell(k-1)-1\ (\ell=3,\dots,k-4) \pmod{L})\\
1&(n\equiv k-4,k-2,2k-5,2k-3,\\
&\hspace{8mm}(k-3)(k-1)-1,(k-3)(k-1)+1,\\
&\hspace{8mm}(k-2)(k-1)-1,(k-2)(k-1)+1,\\
&\hspace{8mm}\ell(k-1)-2,\ell(k-1)\ (\ell=1,\dots,k-2) \pmod{L})\\
0&(\text{otherwise})
\end{cases}.
\end{align}
It follows from \eqref{k2m+10d} and \eqref{k2m+1-1d} that $Z_n$ satisfies the recurrence relation
\begin{equation}
\begin{cases}
Z_n=Z_{n-1}+Z_{n-k+1}-Z_{n-k}+\phi_n\\
Z_{-3}=Z_{-2}=\dots=Z_{k-5}=0,\,Z_{k-4}=1
\end{cases}.
\end{equation}
Solving this recurrence relation, we obtain
\begin{equation}
Z_n=\frac{n^2+5n}{L}+(\text{periodic function with period $L$}).
\end{equation}
Hence, by \eqref{k2m+1-1d}, the degree $d_n$ is given by
\begin{equation}
d_n=\frac{2n^2+(10-2k)n}{L}+g_n,
\end{equation}
where $g_n$ is a periodic function with period $L$.
This formula is proved for $k=5,7,9$.
For odd $k\geq11$, the formula follows provided that the singularity patterns \eqref{k2m+1-1} and \eqref{k2m+10} are correct.
Since $d_n$ is expressed as the sum of a quadratic polynomial in $n$ and a periodic function, the algebraic entropy is zero.
In the case $k=5$, $g_n$ has period $L/2=7$ and is defined by
\begin{equation}
(g_0,g_1,\dots,g_6)=\frac{1}{7}(0, 6, 3, 5, 5, 3, 6).
\end{equation}
For $k=7$, $g_n$ has period $L=32$ and is defined by
\begin{equation}
\begin{aligned}
(g_0,g_1,\dots,g_{31})=\frac{1}{16}(&0, 1, 0, 13, 8, 17, 8, -3, 0, 1, 16, 13, 8, 1, -8, -3, 16, \\&17, 16, -3, -8, 1, 8, 13, 16, 1, 0, -3, 8, 17, 8, 13).
\end{aligned}
\end{equation}
For $k=9$, $g_n$ has period $L=58$ and is defined by
\begin{equation}
\begin{aligned}
(g_0,g_1,\dots,g_{57})=\frac{1}{29}(&0, 3, 4, 3, 0, 24, 17, 37, 26, 13, -2, -19, -9, -1, 34, 38, 40, \\&11, -20, -24, -30, -9, 39, 56, 42, -3, -21, -41, -34, 0, \\&32, 62, 32, 0, -34, -41, -21, -3, 42, 56, 39, -9, -30, -24, \\&{-20}, 11, 40, 38, 34, -1, -9, -19, -2, 13, 26, 37, 17, 24).
\end{aligned}
\end{equation}

\section{Conclusion}
In this article, we introduced generalized $q$-Painlev\'{e} equations  \eqref{y2}, \eqref{y3} with a parameter $k \in \Z_{\ge 4}$.
These $q$-discrete equations are extensions of the $q$-discrete Painlev\'{e} equations in the sense that they coincide with the $q$-Painlev\'{e} I and II equations for $k=4$ and $k=5$ respectively, and that they have the co-primeness property which is regarded as an algebraic reinterpretation of singularity confinement.
We conjecture that the same property holds for $k \ge 6$, and wish to investigate the conjectures \ref{thm1} and \ref{thm2}.
Moreover, based on observations of singularity confinement and degree growth, the resulting Painlevé type equations are expected to be integrable.
Extension of the results in this article to other discrete Painlev\'{e} equations is one of the problems we wish to address in the future.

\subsection*{Acknowledgements}

The author would like to thank M. Kanki, T. Mase, T. Tokihiro and R. Willox for useful comments.
The author would also like to thank T. Mase for his very helpful comments on the contents of Section 5 during the revision of this paper.

\label{lastpage}
\end{document}